\documentclass[a4paper,UKenglish,cleveref,autoref,thm-restate]{lipics-v2021}

\usepackage{mathtools}
\usepackage{bm}
\usepackage{ascmac}
\usepackage{algorithm}
\usepackage{algpseudocodex}

\title{Theoretical Aspects of Generating Instances with Unique Solutions: Pre-assignment Models for Unique Vertex Cover} 

\titlerunning{Pre-assignment Models for
Unique Vertex Cover} 

\author{Takashi Horiyama}{Faculty of Information Science and Technology, Hokkaido University, Japan}{horiyama@ist.hokudai.ac.jp}{https://orcid.org/0000-0001-9451-259X}{JSPS KAKENHI Grant Numbers JP20H05964, JP22H03549}

\author{Yasuaki Kobayashi}{Faculty of Information Science and Technology, Hokkaido University, Japan}{koba@ist.hokudai.ac.jp}{https://orcid.org/0000-0003-3244-6915}{JSPS KAKENHI Grant Numbers JP20H00595, JP23H03344}
\author{Hirotaka Ono
}{Graduate School of Informatics, Nagoya University, Japan 
}{ono@nagoya-u.jp}{https://orcid.org/0000-0003-0845-3947}{JSPS KAKENHI Grant Numbers JP20H05967, JP21K19765, JP22H00513}

\author{Kazuhisa Seto}{Faculty of Information Science and Technology, Hokkaido University, Japan}{seto@ist.hokudai.ac.jp}{}{JSPS KAKENHI Grant Number JP22H03549, JP22H00513}

\author{Ryu Suzuki}{Graduate School of Information Science and Technology, Hokkaido University, Japan }{suzuki-r1991@eis.hokudai.ac.jp}{}{}

\authorrunning{T. Horiyama, et al.} 

\Copyright{Takashi Horiyama, Yasuaki Kobayashi, Hirotaka Ono, Kazuhisa Seto, and Ryu Suzuki} 



\category{} 

\relatedversion{} 

\nolinenumbers 

\EventEditors{John Q. Open and Joan R. Access}
\EventNoEds{2}
\EventLongTitle{42nd Conference on Very Important Topics (CVIT 2016)}
\EventShortTitle{CVIT 2016}
\EventAcronym{CVIT}
\EventYear{2016}
\EventDate{December 24--27, 2016}
\EventLocation{Little Whinging, United Kingdom}
\EventLogo{}
\SeriesVolume{42}
\ArticleNo{23}

\usepackage{xcolor}
\usepackage{xspace}

\newcommand{\pauvc}{{\textsc PAU-VC}\xspace}

\newcommand{\cost}{\mathrm{cost}}
\newcommand{\SC}{\mathbf{SC}}

\newcommand{\bv}{\bm{v}}

\floatname{algorithm}{Procedure}

\def\multiset#1#2{\ensuremath{\left(\kern-.3em\left(\genfrac{}{}{0pt}{}{#1}{#2}\right)\kern-.3em\right)}}

\begin{document}

\maketitle

\begin{abstract}
The uniqueness of an optimal solution to a combinatorial optimization problem attracts many fields of researchers' attention because it has a wide range of applications, it is related to important classes in computational complexity, and an instance with only one solution is often critical for algorithm designs in theory. However, as the authors know, 
there is no major benchmark set consisting of only instances with unique solutions, and no algorithm generating instances with unique solutions is known; 
a systematic approach to getting a problem instance guaranteed having a unique solution would be helpful. A possible approach is as follows: Given a problem instance, we specify a small part of a solution in advance so that only one optimal solution meets the specification. This paper formulates such a ``pre-assignment'' approach for the
vertex cover problem as a typical combinatorial optimization
problem and discusses its computational complexity.
First, we show that the problem is $\Sigma^P_2$-complete in general, while 
the problem becomes NP-complete when an input graph is bipartite. 
We then present an $O(2.1996^n)$-time algorithm for general graphs and an $O(1.9181^n)$-time algorithm for bipartite graphs, where $n$ is the number of vertices. The latter is based on an FPT algorithm with $O^*(3.6791^{\tau})$ time for vertex cover number $\tau$. Furthermore, we show that the problem for trees can be solved in $O(1.4143^n)$ time.
\end{abstract}

\section{Introduction}
\label{sec:intro}
Preparing a good benchmark set is indispensable for evaluating the actual performance of problem solvers, such as SAT solvers, combinatorial optimization solvers, and learning algorithms. 
This is the reason that instance generation is a classical topic in the fields of AI and OR, including optimization and learning. 
Indeed, there are well-known benchmark sets, such as the TSPLIB benchmark set for the traveling salesperson problem~\cite{reinelt1991tsplib}, the UCI Machine Learning Repository dataset for machine learning~\cite{asuncion2007uci}, SATLIB for SAT~\cite{hoos2000satlib}, and other benchmark sets for various graph optimization problems in the DIMACS benchmarks~\cite{DIMACS}. 
In these benchmarks, instances are generated from actual data or artificially generated. 
Instances generated from actual data are closely related to the application, so it is a very instance that should be solved in practice and is therefore suitable as a benchmark. On the other hand, it is not easy to generate sufficient instances because each one is hand-made, and there are confidentiality issues.
Random generation is often used to compensate for this problem and has also been studied for a long time, especially how to generate good instances from several points of view, such ``difficulty'' control and variety. 

In this paper, we consider generating problem instances with unique solutions. It is pointed out that the ``unique solution'' is related to the nature of intractability in SAT. Indeed, PPSZ, currently the fastest SAT algorithm in theory~\cite{scheder2022ppsz}, is based on the unique k-SAT solver. Furthermore, a type of problem that determines whether or not it has a unique solution (e.g., \textsc{Unique SAT}) belongs to a class above NP and coNP~\cite{BlassEt:USAT}, but the status of that class is still not well understood, where an exception is TSP; \textsc{Unique TSP} is known to be $\Delta^p_2$-complete~\cite{Papadimitriou:UOTSP}. 
Thus, it is desirable to have a sufficient number of instances with unique solutions in many cases, such as when problem instances with unique solutions are used as benchmarks or when we evaluate the performance of a solver for the uniqueness check version of combinatorial optimization problems. In other words, instances with unique solutions have an important role in experimental studies of computational complexity theory. 

In this paper, we focus on \textsc{Vertex Cover} problem as a representative of combinatorial optimization problems (also referred Drosophila in the field of parameterized algorithms~\cite{downey2013fundamentals}), propose a ``pre-assignment'' model as an instance generation model for this problem, and consider the computational complexity of instance generation based on this model. One of the crucial points in generating problem instances with unique solutions is that it is difficult for unweighted combinatorial optimization problems to apply approaches used in conventional generation methods. A simple and conventional way of generating instances is a random generation, but the probability that a randomly generated problem instance has a unique solution is low. Furthermore, although some combinatorial problems have randomized algorithms modifying an instance into an instance with a unique solution, they strongly depend on the structure of the problems, and it is non-trivial to apply them to other problems, including \textsc{Vertex Cover}~\cite{valiantV1986NP,MulmuleyVV1987isolation}. 
Besides, the planted model used in SAT instance generation requires guaranteeing the non-existence of solutions other than the assumed solution, which is also a non-trivial task, or some planted instances are shown to be easy to solve~\cite{KrivelevichV06}. Thus, generating instances with unique solutions requires a different approach than conventional instance generation methods. 

A pre-assignment for \emph{uniquification} in a problem means assigning an arbitrary part of the decision variables in the problem so that only one solution is consistent with the assignment. In SAT, for example, we select some of the Boolean variables and assign true or false to each so that only one truth assignment consistent with the pre-assignment satisfies the formula. In the case of \textsc{Vertex Cover}, we 
select some vertices and assign a role to each, i.e.,  a cover vertex or a non-cover vertex, so that only one minimum vertex cover of the graph is consistent with the pre-assignment.  
Such a pre-assignment for uniquification enables the generation of problem instances with unique solutions to combinatorial optimization problems. In the case of \textsc{Vertex Cover}, the graph obtained by the pre-assignment (see \Cref{obs:reduction} for the detail) has only one vertex cover with the minimum size, which can be used as an instance of \textsc{Unique Vertex Cover} problem. 
A strong point of the pre-assignment-based instance generation is that it transforms an instance of an ordinary problem into an instance of the uniqueness check version of the corresponding problem; if the ordinary version of the problem has a rich benchmark instance set, we can expect to obtain the same amount of instance set by applying the pre-assignment models. Indeed, many papers study instance generations of the ordinary combinatorial optimization problem as seen later; they can be translated into the results of instances with unique solutions. 

From this perspective, we formulate the pre-assignment problem to \emph{uniquify} an optimal solution for the vertex cover problem instance. We name the problem \pauvc, 
whose formal definition is given in the Models section. 
We consider three types of scenarios of pre-assignment: \textsc{Include}, \textsc{Exclude}, and \textsc{Mixed}. 
Under these models, this study aims to determine the computational complexity of the pre-assignment for uniquification of \textsc{Vertex Cover}.

\subsection{Our results}
We investigate the complexity of \pauvc for all pre-assignment models. 
As hardness results, we show:
\begin{enumerate}
    \item \pauvc is $\Sigma^P_2$-complete under any model, 
    \item \pauvc for bipartite graphs is NP-complete under any model.  
\end{enumerate}
For the positive side, we design exact exponential-time algorithms and FPT algorithms for the vertex cover number. 
Let $n$ be the number of the vertices and $\tau$ be the vertex cover number of $G$. Then, we present
\begin{enumerate} \setcounter{enumi}{2}
\item an $O(2.1996^n)$-time algorithm, which works for any model, 
\item an $O^*(3.6791^{\tau})$-time algorithm under \textsc{Include} model,\footnote{The $O^*$ notation suppresses polynomial factors in $n$.} 
\item an $O^*(3.6791^{\tau})$-time algorithm under \textsc{Mixed} and \textsc{Exclude} models,  
\end{enumerate}
Due to the FPT algorithms, \pauvc for bipartite graphs can be solved in $O(1.9181^n)$ time. Here, it should be noted that the fourth and fifth algorithms work in different ways, though they have the same running time.

The last results are for trees. Many graph problems are intractable for general graphs but polynomially solvable for trees. Most of such problems are just NP-complete, but some are even PSPACE-complete~\cite{demaine2015linear}. 
Thus, many readers might consider that \pauvc for trees is likely solvable in polynomial time. On the other hand, not a few problems (e.g., \textsc{Node Kayles}) are intractable in general, but the time complexity for trees still remains open, and only exponential-time algorithms are known~\cite{BodlaenderKT15,yoshiwatari2022winnerc}. In the case of \pauvc, no polynomial-time algorithm for trees is currently known. Instead, we give an exponential upper bound of the time complexity of \pauvc for trees. 
\begin{enumerate}\setcounter{enumi}{5}
\item an $O(1.4143^n)$-time algorithm for trees under \textsc{Include} model, and 
\item an $O(1.4143^n)$-time algorithm for trees under \textsc{Mixed} and \textsc{Exclude} models. 
\end{enumerate}
These algorithms also work differently, though they have the same running time.

\smallskip 

In the context of instance generation, these results imply that the pre-assignment approach could not be universally promising but would work well to generate instances with small vertex covers.  
\subsection{Related work}
Instance generation in combinatorial (optimization) problems is well-studied from several points of view, such as controlling some attributes~\cite{asahiro1996random} and hard instance generation~\cite{horie1997hard}, and some empirical studies use such generated instances for performance evaluation~\cite{cha1995performance}. 
Such instances are desirable to be sufficiently hard because they are supposed to be used for practical performance evaluation of algorithms for problems considered hard in computational complexity theory (e.g., ~\cite{sanchis1995generating,neuen2017benchmark,mccreesh2018subgraph}). 
Other than these, many studies present empirically practical instance generators (e.g., \cite{ullrich2018generic}). From the computational complexity side, it is shown that, unless NP $=$ coNP, for most NP-hard problems, there exists no polynomial-time algorithm capable of generating all instances of the problem, with known answers~\cite{sanchis1990complexity,matsuyama2021hardness}. 

A natural application of pre-assignment for uniquification is puzzle instance generation. In pencil puzzles such as \textsc{Sudoku}, an instance is supposed to have a unique solution as the answer. 
Inspired by them, Demaine et al.~\cite{DemaineEt:FCP} define FCP (Fewest Clues Problem) type problems including FCP-SAT and FCP-\textsc{Sudoku}. 
The FCP-SAT is defined as follows: Let $\phi$ be a CNF formula with a set of boolean variables $X$. Consider a subset $Y \subseteq X$ of variables and a partial assignment $f_Y : Y \to \{0, 1\}$.
If there is a unique satisfying assignment $f_X: X \to \{0, 1\}$ extending $f_Y$, that is, $f_Y(x) = f_X(x)$ for $x \in Y$, we call variables in $Y$ \emph{clues}.
The FCP-SAT problem asks, given a CNF formula $\phi$ and an integer $k$, whether $\phi$ has a unique satisfying assignment with at most $k$ clues.
They showed that FCP 3SAT, the FCP versions of several variants of the SAT problem, 
including \textsc{FCP 1-in-3 SAT}, are $\Sigma^P_2$-complete. 
Also FCP versions of several pencil-and-paper puzzles, including \textsc{Sudoku} are $\Sigma^P_2$-complete. 
Since the setting of FCP is suitable for puzzles, FCP versions of puzzles are investigated~\cite{HiguchiK19NP,Goergen2022All}.

The uniqueness of an (optimal) solution itself has been intensively and extensively studied for many combinatorial optimization problems. Some of them are shown to have the same time complexity to \textsc{Unique SAT}~\cite{hudry2019complexity,DBLP:journals/tcs/HudryL19}, while \textsc{Unique TSP} is shown to be $\Delta^P_2$-complete~\cite{Papadimitriou:UOTSP}. Juban \cite{Juban99} provides the dichotomy theorem for checking whether there
exists a unique solution to a given propositional formula. The theorem partitions the instances of the
problem between the polynomial-time solvable and coNP-hard cases, where Horn, anti-Horn, affine, and 2SAT formulas are the only polynomial-time solvable cases.


\section{Preliminaries}\label{sec:preli}
\subsection{Notations}
Let $G$ be an undirected graph.
Let $V(G)$ and $E(G)$ denote the vertex set and edge set of $G$, respectively.
For $v \in V(G)$, we denote by $N_G(v)$ the set of neighbors of $v$ in $G$, i.e., $N_G(v) = \{w \in V(G) : \{v, w\} \in E(G)\}$.
We extend this notation to sets: $N_G(X) = \bigcup_{v \in X}N_G(v) \setminus X$.
We may omit the subscript $G$ when no confusion is possible.
For $X \subseteq V(G)$, the subgraph of $G$ obtained by deleting all vertices in $X$ is denoted by $G - X$ or $G\setminus X$.

A \emph{vertex cover} of $G$ is a vertex set $C \subseteq V(G)$ such that for every edge in $G$, at least one end vertex is contained in $C$.
We particularly call a vertex cover of $G$ with cardinality $k$ a \emph{$k$-vertex cover} of $G$.
The minimum cardinality of a vertex cover of $G$ is denoted by $\tau(G)$. 
A set $I \subseteq V(G)$ of vertices that are pairwise non-adjacent in $G$ is called an \emph{independent set} of $G$.
It is well known that $I$ is an independent set of $G$ if and only if $V(G) \setminus I$ is a vertex cover of $G$.
A \emph{dominating set} of $G$ is a vertex subset $D \subseteq V(G)$ such that $V(G) = D \cup N(D)$, that is, every vertex in $G$ is contained in $D$ or has a neighbor in $D$.
If a dominating set $D$ is also an independent set of $G$, we call it an \emph{independent dominating set} of $G$.


\subsection{Models}\label{sec:model}
In \pauvc, given a graph $G$ and an integer $k$, the goal is to determine whether $G$ has a ``feasible pre-assignment of size at most $k$''.
In this context, there are three possible models to consider: \textsc{Include}, \textsc{Exclude}, and \textsc{Mixed}.
Under \textsc{Include} model, a pre-assignment is defined as a vertex subset $U \subseteq V(G)$, and a pre-assignment $\tilde{U}$ is said to be \emph{feasible} if there is a minimum vertex cover $U^*$ of $G$ (i.e., $\tau(G) = |U^*|$) such that $U \subseteq U^*$ and for every other minimum vertex cover $U$ of $G$, $\tilde{U} \setminus U$ is nonempty.
In other words, $\tilde{U}$ is feasible if there is a unique minimum vertex cover $U^*$ of $G$ that includes $\tilde{U}$ (i.e., $\tilde{U} \subseteq U^*$).
Under \textsc{Exclude} model, a pre-assignment is also defined as a vertex subset $\tilde{U} \subseteq V(G)$ as well as \textsc{Include}, and $\tilde{U}$ is said to be feasible if there is a unique minimum vertex cover $U^*$ of $G$ that excludes $X$ (i.e., $U^* \subseteq V(G) \setminus \tilde{U}$).
For these two models, the size of a pre-assignment is defined as its cardinality.
Under \textsc{Mixed} model, a pre-assignment is defined as a pair of vertex subsets $(\tilde{U}_{\rm in}, \tilde{U}_{\rm ex})$, and $(\tilde{U}_{\rm in}, \tilde{U}_{\rm ex})$ is said to be feasible if there is a unique minimum vertex cover $U^*$ of $G$ that includes $\tilde{U}_{\rm in}$ and excludes $\tilde{U}_{\rm ex}$ (i.e., $\tilde{U}_{\rm in} \subseteq U^* \subseteq V(G) \setminus \tilde{U}_{\rm ex}$).
We can assume that $\tilde{U}_{\rm in}$ and $\tilde{U}_{\rm ex}$ are disjoint, as otherwise the pre-assignment is trivially infeasible.
The size of the pre-assignment $(\tilde{U}_{\rm in}, \tilde{U}_{\rm ex})$ is defined as $|\tilde{U}_{\rm in}| + |\tilde{U}_{\rm ex}|$ under the \textsc{Mixed} model.
\begin{observation}\label{obs:reduction}
    Let $G$ be a graph and let $(\tilde{U}_{\rm in}, \tilde{U}_{\rm ex})$ be a feasible pre-assignment of $G$.
    Then, $G-(\tilde{U}_{\rm in} \cup N(\tilde{U}_{\rm ex}))$ has a unique minimum vertex cover of size $\tau(G) - |\tilde{U}_{\rm in} \cup N(\tilde{U}_{\rm ex})|$.
\end{observation}

In the following, we compare these three models.

\begin{theorem}\label{thm:include-and-exclude}
    If $G$ has a feasible pre-assignment of size at most $k$ in the \textsc{Include} model, then $G$ has a feasible pre-assignment of size at most $k$ under the \textsc{Exclude} model.
\end{theorem}

\begin{proof}
    Let $\tilde{U} \subseteq V(G)$ be a feasible pre-assignment for $G$ in the \textsc{Include} model.
    Let $U^* \subseteq V(G)$ be the unique minimum vertex cover of $G$ such that $\tilde{U} \subseteq U^*$.
    Observe that $N(v) \setminus U^* \neq \emptyset$ for $v \in U^*$.
    This follows from the fact that if $N(v) \subseteq U^*$, $U^* \setminus \{v\}$ is also a vertex cover of $G$, which contradicts the minimality of $U^*$.
    For $v \in \tilde{U}$, we let $v'$ be an arbitrary vertex in $N(v) \setminus U^*$ and define $\tilde{U}' \coloneqq \{v' : v \in \tilde{U}\}$.
    Note that $v'$ and $w'$ may not be distinct even for distinct $v, w \in \tilde{U}$.
    We claim that $\tilde{U}'$ is a feasible pre-assignment of $G$ in the \textsc{Exclude} model.
    
    As $v' \notin U^*$ for each $v \in \tilde{U}$, $U^* \subseteq V(G) \setminus \tilde{U}'$ holds.
    To see the uniqueness of $U^*$ (under the pre-assignment $\tilde{U}'$ in \textsc{Exclude}), suppose that there is a minimum vertex cover $U$ of $G$ with $U \neq U^*$ such that $U \subseteq V(G) \setminus \tilde{U}'$.
    For each $v' \in \tilde{U}'$, all vertices in $N(v')$ must be included in $U$, which implies that $v \in U$.
    Thus, we have $\tilde{U} \subseteq U$, contradicting the uniqueness of $U^*$. By $|\tilde{U}'| \le |\tilde{U}|$, the theorem holds.
\end{proof}

The converse of \Cref{thm:include-and-exclude} does not hold in general.
Let us consider a complete graph $K_n$ with $n$ vertices.
As $\tau(K_n) = n - 1$, exactly one vertex is not included in a minimum vertex cover of $K_n$.
Under \textsc{Exclude} model, a pre-assignment containing exactly one vertex is feasible for $K_n$.
However, under \textsc{Include} model, any pre-assignment of at most $n - 2$ vertices does not uniquify a minimum vertex cover of $K_n$.
This indicates that \textsc{Exclude} model is ``stronger'' than \textsc{Include} model.
The theorem below shows that \textsc{Exclude} model and \textsc{Mixed} model are ``equivalent''.

\begin{theorem}\label{thm:mix-and-exclude}
    A graph $G$ has a feasible pre-assignment of size at most $k$ under the \textsc{Exclude} model if and only if $G$ has a feasible pre-assignment of size at most $k$ under the \textsc{Mixed} model.
\end{theorem}
\begin{proof}
    The forward implication follows from the observation that for any feasible pre-assignment $\tilde{U}$ of $G$ in the \textsc{Exclude} model, the pair $(\emptyset, \tilde{U})$ is a feasible pre-assignment of $G$ in the \textsc{Mixed} model.
    Thus, we consider the converse implication.
    
    Let $(\tilde{U}_{\rm in}, \tilde{U}_{\rm ex})$ be a feasible pre-assignment of $G$ under the \textsc{Mixed} model.
    Let $U^* \subseteq V(G)$ be the unique vertex cover of $G$ such that $\tilde{U}_{\rm in} \subseteq U^* \subseteq V(G) \setminus \tilde{U}_{\rm ex}$.
    As $U^*$ is a minimum vertex cover of $G$, we have $N(v) \setminus U^* \neq \emptyset$ for $v \in U^*$.
    For $v \in \tilde{U}_{\rm in}$, we let $v'$ be an arbitrary vertex in $N(v) \setminus U^*$ and define $\tilde{U} \coloneqq \tilde{U}_{\rm ex} \cup \{v' : v \in \tilde{U}_{\rm in}\}$.
    We claim that $\tilde{U}$ is a feasible pre-assignment of $G$ under the \textsc{Exclude} model.
    
    As $v' \notin U^*$ for each $v \in \tilde{U}_{\rm in}$, it holds that $U^* \subseteq V(G) \setminus \tilde{U}$.
    To see the uniqueness of $U^*$, suppose that there is a minimum vertex cover $U$ of $G$ with $U \neq U^*$ such that $U \subseteq V(G) \setminus \tilde{U}$.
    For each $v \in \tilde{U}$, all vertices in $N(v)$ must be included in $U$.
    This implies that $\tilde{U}_{\rm in} \subseteq U$.
    Thus, we have $\tilde{U}_{\rm in} \subseteq U \subseteq V(G) \setminus \tilde{U}_{\rm ex}$, contradicting the uniqueness of $U^*$.
\end{proof}

\subsubsection{Basic observations}
A class $\mathcal G$ of graphs is said to be \emph{hereditary} if for $G \in \mathcal G$, every induced subgraph of $G$ belongs to $\mathcal G$.
\begin{lemma}\label{lem:unique_alg}
    Let $\mathcal G$ be a hereditary class of graphs.
    Suppose that there is an algorithm $\mathcal A$ for computing a minimum vertex cover of a given graph $G \in \mathcal G$ that runs in time $T(n, m)$, where $n = |V(G)|$ and $m = |E(G)|$.\footnote{We assume that $n + m \le T(n, m) \le T(n', m')$ for $n \le n'$ and $m \le m'$.}
    Then, we can check whether $G \in \mathcal G$ has a unique minimum vertex cover in time $O(\tau(G) \cdot T(n, m))$.
\end{lemma}
\begin{proof}
    We first compute an arbitrary minimum vertex cover $U$ of $G$.
    If $U$ is not a unique minimum vertex cover of $G$, $G$ has another minimum vertex cover $U'$ not containing $v$ for some $v \in U$.
    As $v \notin U'$, we have $N(v) \subseteq U'$.
    To find such a minimum vertex cover for each $v \in U$, it suffices to compute a vertex cover of $G - (N(v) \cup \{v\})$ of size $\tau(G) - |N(v)|$, which can be done by using $\mathcal A$ as $G - (N(v) \cup \{v\}) \in \mathcal G$.
\end{proof}

Using~\Cref{lem:unique_alg}, we have the following corollary.
\begin{corollary}\label{cor:unique_alg}
    Let $G$ be a hereditary class of graphs.
    Suppose that there is an algorithm for computing a minimum vertex cover of a given graph $G \in \mathcal G$ that runs in time $T(n, m)$, where $n = |V(G)|$ and $m = |E(G)|$.
    Given a graph $G \in \mathcal G$ and vertex sets $\tilde{U}_{\rm in}, \tilde{U}_{\rm ex} \subseteq V(G)$, a pre-assignment $(\tilde{U}_{\rm in}, \tilde{U}_{\rm ex})$ is feasible for $G$ under the \textsc{Mixed} model in time $O(\tau(G)\cdot T(n, m))$.
\end{corollary}
\begin{proof}
    If either $\tilde{U}_{\rm in} \cap \tilde{U}_{\rm ex} \neq \emptyset$ or $\tilde{U}_{\rm ex}$ is not an independent set of $G$, the pre-assignment is trivially infeasible.
    Suppose otherwise.
    Observe that $G$ has a minimum vertex cover $U$ with $\tilde{U}_{\rm in} \subseteq C \subseteq V(G) \setminus \tilde{U}_{\rm ex}$ if and only if $G - (\tilde{U}_{\rm in} \cup N_G(\tilde{U}_{\rm ex}))$ has a vertex cover of size $\tau(G) - |\tilde{U}_{\rm in} \cup N_G(\tilde{U}_{\rm ex})|$.
    Therefore, we can check whether $G$ has a unique minimum vertex cover $U$ satisfying $\tilde{U}_{\rm in} \subseteq U \subseteq V(G) \setminus \tilde{U}_{\rm ex}$ using the algorithm in~\Cref{lem:unique_alg} as well.
\end{proof}

\section{Complexity}\label{sec:complexity}
This section is devoted to proving the complexity of \pauvc. In particular, we show that \pauvc is $\Sigma^P_2$-complete on general graphs and NP-complete on bipartite graphs.

\subsection{General graphs}\label{sec:complexity:general}

\begin{theorem}\label{thm:sigma}
{\rm PAU-VC} is $\Sigma^P_2$-complete under any type of pre-assignment.   
\end{theorem}


Recall that a decision problem belongs to $\Sigma^P_2$ if and only if it is solvable in polynomial time by a non-deterministic Turing machine augmented by an oracle for an NP-complete problem. We then see that \pauvc belongs to $\Sigma^P_2$, because  
under a pre-assignment as a witness we can decide in polynomial time whether an input graph has a unique vertex cover with an NP-oracle as shown in \Cref{cor:unique_alg}.   

We then prove the hardness by reducing $\Sigma^P_2$-complete problem \textsc{FCP 1-in-3 SAT} (\cite{DemaineEt:FCP}) to \pauvc.
Here, we utilize the reduction from \textsc{Unique 1-in-3 SAT} to \textsc{Unique VC} shown in~\cite{hudry2019complexity}. Their proof actually reduces an instance of \textsc{1-in-3 SAT} to an instance of \textsc{Vertex Cover}. 
Suppose an instance of \textsc{1-in-3 SAT} 
consists of a set $C$ of $m$ clauses over a set $X$ of $n$ variables. Each clause contains exactly three literals. Here, $X$ also represents the set of positive literals of variables in $X$ and $\bar{X}$ denotes the set of negative literals of variables in $X$.
The goal of \textsc{1-in-3 SAT} is to find a truth assignment of variables such that for each clause in $C$, exactly one literal is true in this assignment.
We call such an assignment a \emph{1-in-3 assignment} (for $C$).
The reduction of \cite{hudry2019complexity} constructs graph $G_C=(V(G_C),E(G_C))$ as an instance of \textsc{Vertex Cover} in the following way. The vertex set $V(G_C)$ consists of three sets of vertices as $V(G_C) \coloneqq V^{X}_{C} \cup V^{X'}_{C} \cup V^{\mathrm{clause}}_{C}$.   
\begin{itemize}
    \item $V^{X}_{C}$ is the set of vertices corresponding to literals in $X$, 
    i.e., $V^{X}_{C} = \bigcup_{x\in X}\{v(x), v(\bar{x})\}$. 
    \item For $V^{X}_{C}$, we prepare the auxiliary vertex set $V^{X'}_{C}=\{r\} \cup \bigcup_{x\in X} \{v'(x),v'(\bar{x}),u(x)\}$. As we explain later, for each variable $x$, three vertices $v'(x),v'(\bar{x}),u(x)$ form a triangle, and $v'(x)$ and $v'(\bar{x})$ are associated with $v(x)$ and $v(\bar{x})$, respectively. Vertex $r$ is a special vertex connected with $u(x)$ for every $x\in X$.     
    \item For a clause $c_i$, we prepare three vertices $v_{i}^{(1)},v_{i}^{(2)}$, and $v_{i}^{(3)}$, where $v_{i}^{(j)}$ is associated with the $j$-th literal in $c_i$. 
    Let $V^{\mathrm{clause}}_{C} = \bigcup_{c_i\in C}\{v_{i}^{(1)},v_{i}^{(2)},v_{i}^{(3)}\}$. 
    \item For these vertices, we define the set of edges: $E(G_C)\coloneqq E^{X}_{C} \cup E^{\mathrm{aux}}_{C} \cup E^{\mathrm{clause}}_{C} \cup E^{\mathrm{inter}}_{C}$, where 
    \begin{itemize}
    \item $E^{X}_{C}\coloneqq \bigcup_{x\in X}\{\{v(x),v(\bar{x})\},\{\{v(x),v'(\bar{x})\}, \{v(\bar{x}),v'({{x}})\}\}$ 
    \item $E^{\mathrm{aux}}_{C}\coloneqq \bigcup_{x\in X} \{\{v'(x),v'(\bar{x})\},\{v'(x),u(x)\}, \{v'(\bar{x}),u(x)\}, $\\
    $\{u(x),r\}\}$ 
     \item $E^{\mathrm{clause}}_{C}\coloneqq\bigcup_{c_i\in C} \{\{v_{i}^{(1)},v_{i}^{(2)}\}, \{v_{i}^{(2)},v_{i}^{(3)}\},\{v_{i}^{(3)},v_{i}^{(1)}\}\}$    
     \item $E^{\mathrm{inter}}_{C}\coloneqq \bigcup_{c_i\in C} \{\{v_{i}^{(j)},v(l_{i}^{(j)})\} \mid j=1,2,3\} \cup$\\$\bigcup_{c_i\in C} \{\{v(\bar{l}_i^{(1)}),v(\bar{l}_i^{(2)})\},\{v(\bar{l}_i^{(2)}),v(\bar{l}_i^{(3)})\}$, $\{v(\bar{l}_i^{(3)}),v(\bar{l}_i^{(1)})\}\}$, where $l_{i}^{(j)}$ represents the $j$-th literal of clause $c_i$. For example, for $c_4=(x_5, \bar{x}_6, x_7)$, $l_4^{(1)}\coloneqq x_5$, $l_4^{(2)}\coloneqq \bar{x}_6$, and $l_4^{(3)}\coloneqq x_7$. 
     Edge $\{v_{4}^{(1)},v(l_{4}^{(1)})\}$ actually stands for edge $\{v_{4}^{(1)},v(x_5)\}$, and $\{v(\bar{l}_4^{(1)}),v(\bar{l}_4^{(2)})\}$ actually stands for edge $\{v(\bar{x}_5),v(x_6)\}$.  
    \end{itemize}
\end{itemize}
\begin{figure}
    \centering
    \includegraphics[scale=0.4]{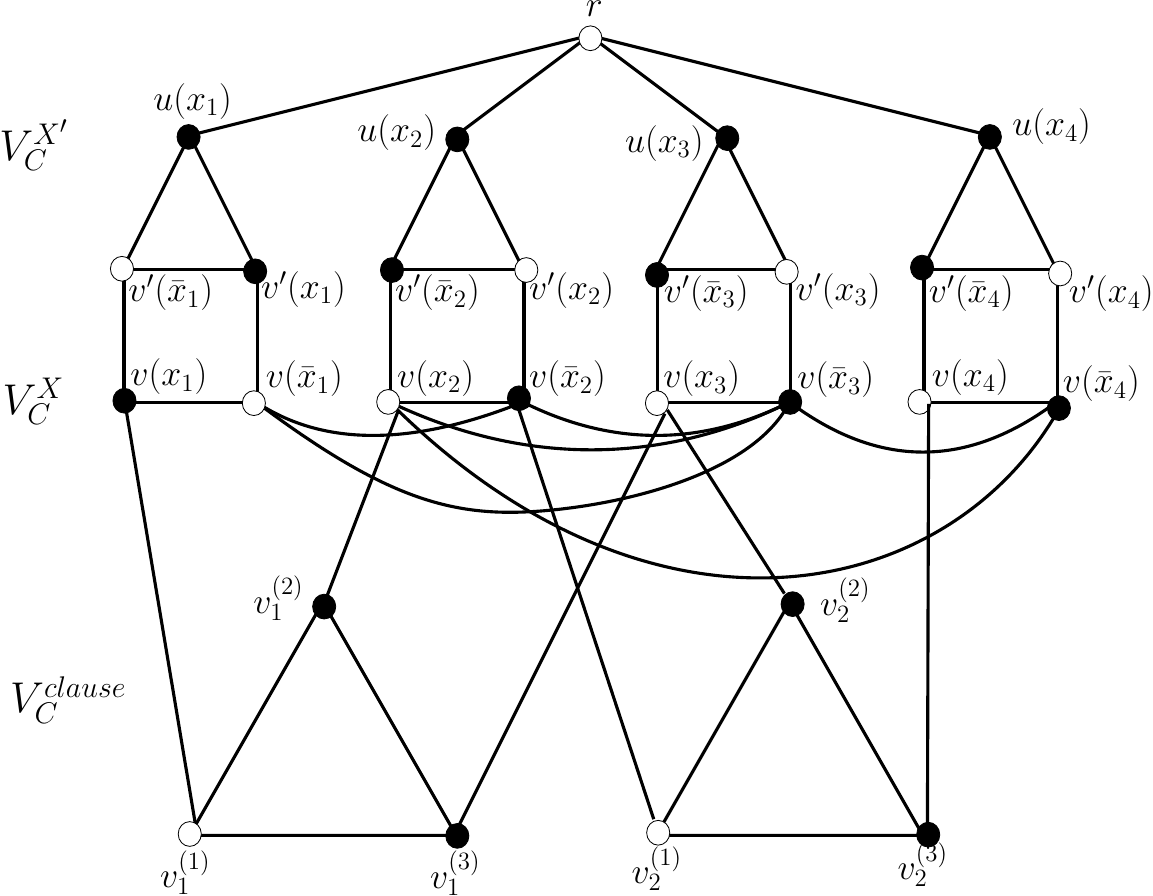}
    \caption{An example of the reduction from \textsc{FCP 1-in-3 SAT} to \pauvc. 
    The instance of \textsc{FCP 1-in-3 SAT} has four variables and two clauses, 
    $c_1 = (x_1, x_2, x_3), c_2 = (\bar{x}_2, x_3, x_4)$. The sixteen black vertices form the an optimal vertex cover $U_f$, which corresponds to the truth assignment $(x_1,x_2,x_3,x_4)=(1,0,0,0)$ satisfying 1-in-3 condition.}
    \label{fig:USAT-reduction}
\end{figure}
Figure \ref{fig:USAT-reduction} shows an example of the reduction. As seen in Figure \ref{fig:USAT-reduction}, $G_C$ is partitioned into variable gadgets, clause gadgets, and $r$. Each component of the variable gadgets forms so-called a house graph, and each component of the clause gadgets forms a triangle.  

\medskip 

We show that there is a bijection from the solutions of instance $C$ of \textsc{1-in-3 SAT} to the optimal solutions of $G_C$ of \textsc{Vertex Cover}, which are vertex covers with size $k=3n+2m$. The proof is constructive; we give an injective and surjective mapping. Note that we can forget the uniqueness of a solution at this moment. 

We first give an injection $\phi$ from solutions of instance $C$ of \textsc{1-in-3 SAT} to vertex covers of $G_C$ with size $k=3n+2m$.   
Let $f$ be a solution of 1-in-3 instance $C$; $f$ is a truth assignment of variables on $X$. Our injection $\phi$ maps $f$ into vertex cover $U_f$ of $G_C$.  
We first include every $u(x)$ in $U_f$.
According to $f$, we decide whether each vertex in $V^{X}_{C} \cup V^{X'}_{C}$ (i.e., variable gadgets) is included in the minimum vertex cover $U_f$ or not: $v(x),v'({x})\in U_f$ and $v(\bar{x}),v'(\bar{x})\not\in U_f$ if $f(x)=1$, $v(x),v'({x})\not\in U_f$ and $v(\bar{x}),v'(\bar{x})\in U_f$ otherwise. In other words, ``true'' vertices in $V^{X}_{C}$ and ``false'' vertices in $V^{X'}_{C}$ are included in $U_f$. Note that this definition guarantees that the mapping is injective.  
We can see that the number of vertices in $V^{X}_{C} \cup V^{X'}_{C}$ included in $U_f$ is $3n$, and every edge in $E^{X}_{C} \cup E^{\mathrm{aux}}_{C}$ is covered by a vertex in $U_f$. 

We next see triangles in the clause gadgets. Since $f$ is a 1-in-3 assignment of $C$, each triangle for $c_i$ is connected with exactly one vertex $v_i^{(j)}$ of $V^{X}_{C}$ included in $U_f$. We then exclude $v_i^{(j)}$ from $U_f$ and include the two other $v_i^{(j')}$ and $v_i^{(j'')}$ in $U_f$, which cover all edges in the triangle. 
In other words, if a vertex in such a triangle is adjacent to a false (resp., true) vertex in $V^{X}_{C}$, it is in $U_f$ (resp., not in $U_f$). The number of vertices here included in $U_f$ is $2m$, and the size of $U_f$ is $3n+2m$; this completes the injection. 

Now we see that $U_f$ is a vertex cover of $G_C$. We have already seen that the edges in $E^{X}_{C} \cup E^{\mathrm{aux}}_{C} \cup E^{\mathrm{clause}}_{C}$ are covered by $U_f$. 
The remaining are the edges in $E^{\mathrm{inter}}_{C}$. These edges are between $V^{X}_{C}$ and $V^{\mathrm{clause}}_{C}$ or between two vertices in $V^{X}_{C}$. 
If an edge in the former category is incident with a true vertex in $V^{X}_{C}$, it is covered by the true vertex. If an edge in the former category is incident with a false vertex in $V^{X}_{C}$, the other endpoint in a triangle is included in $U_f$; it is again covered. An edge in the latter category is part of a triangle of complement literals in $c_i$. Since $f$ is a 1-in-3 assignment, the triangle contains two true vertices in $V^{X}_{C}$; every edge in the triangle is covered. By these, every edge in $G_C$ is covered by a vertex in $U_f$.  

We next show that the above mapping is surjective. To this end, we characterize a $(3n+2m)$-vertex cover of $G_C$.  
\begin{claim}
    If $G_C$ has a vertex cover $U$ with $3n+2m$, $U$ satisfies the following: 
    \begin{enumerate}
        \item for every $c_i\in C$, $|\{v_i^{(1)}, v_i^{(2)}, v_i^{(3)}\}\cap U|=2$,
        \item $r\not\in U$ and $u(x)\in U$ for every $x\in X$, 
        \item for every $x\in X$, $\{v(x),v(\bar{x}),v'(x),v'(\bar{x})\}\cap U$ is either $\{v(x),v'({x})\}$ or $\{v(\bar{x}),v'(\bar{x})\}$,  
        \item $U$ does not contain any vertex in $V_C^X$ adjacent to a vertex in $V_C^{\mathrm{clause}}\cap U$.
    \end{enumerate}
\end{claim}
\begin{proof}
Suppose $U$ is a $(3n+2m)$-vertex cover of $G_C$. 
We first consider $E^{\mathrm{clause}}_{C}$. To cover the edges 
$\{v_{i}^{(1)},v_{i}^{(2)},\}, \{v_{i}^{(2)},v_{i}^{(3)}\},\{v_{i}^{(3)},v_{i}^{(1)}\}$ for every $c_i\in C$, at least two of $v_{i}^{(1)},v_{i}^{(2)},v_{i}^{(3)}$ should be included in $U$. Thus for every $c_i\in C$, $|\{v_i^{(1)}, v_i^{(2)}, v_i^{(3)}\}\cap U|\ge 2$ holds, and at least $2m$ vertices in $V^{\mathrm{clause}}_{C}$ are in $U$. Similarly, to cover edges $\{v(x),v(\bar{x})\},\{v(\bar{x}),v'({{x}})\}$, 
 $\{v'({{x}}),v'({\bar{x}})\}$, and $\{v(x),v'(\bar{x})\}$ in $E^{X}_{C} \cup E^{\mathrm{aux}}_{C}$, $U$ should include at least $\{v(x),v'({x})\}$ or $\{v(\bar{x}),v'(\bar{x})\}$. Furthermore, to cover edges $\{v'(x),u(x)\}$ and $\{v'(\bar{x}),u(x)\}$ in $E^{\mathrm{aux}}_{C}$, one more vertex should be included, i.e., if $\{v(x),v'({x})\}\subseteq U$, either $v'(\bar{x})$ or $u(x)$ should be included in $U$, and if $\{v(\bar{x}),v'(\bar{x})\}\subseteq U$, either $v'({x})$ or $u(x)$ should be included in $U$. Hence, for each $x$, at least three vertices in variable gadgets should be included in $U$, and thus at least $3n$ vertices in $V^{X}_C\cup V^{X'}_C$ are in $U$.  Therefore, at least $3n+2m$ vertices are necessary to cover the edges in 
  $E^{X}_{C} \cup E^{\mathrm{aux}}_{C} \cup E^{\mathrm{clause}}_{C}$. Since the size of a vertex cover is $3n+2m$, every inequality should hold as an equality. Thus we have 
  condition \textbf{1}. 

  We then consider edges forming $\{u(x),r\}$. To cover these edges, either $u(x)$ or $r$ should be included in $U$. However, we cannot include in $U$ any vertex other than the ones included in the above argument. Thus, each $\{u(x),r\}$ is covered by $u(x)$, which also implies conditions~\textbf{2}~and~\textbf{3}.

  We now consider condition \textbf{4}. 
  We focus on a clause $c_i=(l_i^{(1)},l_i^{(2)},l_i^{(3)})$. By condition~\textbf{1}, exactly one vertex $v_i^{(j)}$ does not belong to $U$. Without loss of generality, we set $j=1$, i.e., $v_i^{(1)}\not \in U$ and $v_i^{(2)}, v_i^{(3)} \in U$. Since edge $\{v_i^{(1)},v(l_i^{(1)})\}$ should be covered, $U$ contains $v(l_i^{(1)})$, which implies $v(\bar{l}_i^{(1)}) \not\in U$ by condition \textbf{3}. 
  Then, to cover edges $\{v(\bar{l}_i^{(1)}),v(\bar{l}_i^{(2)})\}$, $\{v(\bar{l}_i^{(2)}),v(\bar{l}_i^{(3)})\}$, $\{v(\bar{l}_i^{(3)}),v(\bar{l}_i^{(1)})\}$, 
  $U$ should contain both $v(\bar{l}_i^{(2)})$ and $v(\bar{l}_i^{(3)})$, which implies $v({l}_i^{(2)}), v({l}_i^{(3)})\not\in U$ by condition \textbf{3}. This leads to condition \textbf{4}. 
\end{proof}

\medskip 

Suppose that $U$ is a vertex cover satisfying the conditions of the above claim.  
We then define $f_U$ to satisfy that $f_U(x)=1$ if and only if $v(x) \in U$, which is consistent due to condition \textbf{3} of the claim. Then, we can easily see that the vertex cover mapped from $\phi(f_U)$ is $U$ itself. Thus, there is a bijection from the solutions of instance $C$ of \textsc{1-in-3 SAT} to $(3n+2m)$-vertex covers of $G_C$. 

\medskip 

Now we are ready to show that \textsc{FCP 1-in-3 SAT} can be solved via \pauvc.
In \textsc{FCP 1-in-3 SAT}, given an instance $(X, C)$ of \textsc{1-in-3 SAT} and an integer $k$, the goal is to compute a subset $Y \subseteq X$ of at most $k$ variables such that by taking some truth assignment $f_Y$ to $Y$, there is a unique 1-in-3 assignment $f_X$ to $X$ extending $f_Y$, that is, $f_X(x) = f_Y(x)$ for $x \in Y$.
We call variables in $Y$ \emph{clues}.
There are three models of pre-assignment: \textsc{Include}/\textsc{Exclude}/\textsc{Mixed}. 
As seen in \Cref{thm:mix-and-exclude}, \textsc{Mixed} is essentially equivalent to \textsc{Exclude}. 
Thus we consider only \textsc{Include} and \textsc{Exclude}.  

We first consider the \textsc{Include} model. For a given instance $C$ of \textsc{FCP 1-in-3 SAT}, we construct $G_C$. We claim that $C$ has at most $k$ clues (i.e., a truth assignment for at most $k$ variables) such that the resulting $C$ has a unique solution if and only if $G_C$ has a feasible pre-assignment of size at most $k$. For the if-direction, suppose $X'$ is the set of $k$ variables. 
If $x\in X'$ is assigned true (resp., false), we include $v(x)$ (resp., $v(\bar{x})$) in $\tilde{U}$, which makes a pre-assignment with size $k$. The bijection shown above guarantees that a $(3n+2m)$-vertex cover of $G_C$ containing $\tilde{U}$ is unique, that is, the pre-assignment is feasible.

For the only-if direction, suppose that $\tilde{U}$ is a feasible pre-assignment for $G_C$ with size at most $k$.
Let $U^*$ denote the unique minimum vertex cover of $G_C$ with $\tilde{U} \subseteq U^*$. Note that $\tilde{U}$ should be consistent with the conditions in the claim; for example, $\tilde{U}$ cannot include $v(x)$ and $v(\bar{x})$ simultaneously. 
If $\tilde{U}\subseteq V^X_C$, $\tilde{U}$ defines $k$ clues for $C$. Namely, if $v(x)\in \tilde{U}$ (resp., $v(\bar{x})\in \tilde{U}$), $x$ is assigned true (resp., false) in the clues, whose size becomes $k$.  Then the bijection guarantees the uniqueness of the solution under the clues.  
We then consider the case when $\tilde{U}$ contains a vertex outside of $V^X_C$. Such a vertex can be one in $V_C^{X'}$ or $V_C^{\mathrm{clause}}$. If $u(x)\in \tilde{U}\cap V_C^{X'}$, we can remove it by condition \textbf{2} of the claim. 
If $v'(l)\in \tilde{U}\cap V_C^{X'}$, where $l$ is a literal (i.e., $l=x$ or $l=\bar{x}$ for some $x$), condition \textbf{3} of the claim allows us to replace $\tilde{U}$ by $\tilde{U}\coloneqq\tilde{U}\setminus \{v'(l)\}\cup \{v(l)\}$, which does not increase the size of $\tilde{U}$. By repeatedly applying these procedures, we can assume that $\tilde{U}\cap V_C^{X'}=\emptyset$. The last case is that some $v_i^{(j)} \in \tilde{U} \cap V_C^{\mathrm{clause}}$. Let $l$ be the literal corresponding to $v_i^{(j)}$.  
Then, condition \textbf{4} implies that $v(l)\not\in U^*$ and $v(\bar{l})\in U^*$. 
Conversely, $v(\bar{l}) \in U^*$ implies that $v(l)\not\in U^*$ and $v_i^{(j)}\in U^*$; we can replace $\tilde{U}$ by $\tilde{U}\coloneqq\tilde{U}\setminus \{v_i^{(j)}\}\cup \{v(\bar{l})\}$, which does not increase the size of $\tilde{U}$ again. Thus, all the cases are reduced to $\tilde{U}\subseteq V^X_C$. 
Overall, the if-and-only-if relation holds for the \textsc{Include} model.  

We next consider the \textsc{Exclude} model. 
Also for this case, we show that $C$ has at most $k$ clues (i.e., a truth assignment for at most $k$ variables) such that the resulting $C$ has a unique solution if and only if $G_C$ has a feasible pre-assignment of size at most $k$. For the if-direction, suppose $X'$ is the set of $k$ variables. 
If $x\in X'$ is assigned true (resp., false), we exclude $v(\bar{x})$ (resp., $v({x})$) in $\hat{U}$, which makes a pre-assignment with size $k$. For this case also, the bijection shown above guarantees that $G_C$ has a unique $(3n+2m)$-vertex cover of $G_C$ that does not contain any vertex in $\hat{U}$, that is, the pre-assignment is feasible.

For the only-if direction, suppose that $\hat{U}$ is a feasible pre-assignment for $G_C$. Let $U^*$ denote the unique minimum vertex cover of $G_C$ with $U^* \subseteq V(G_C) \setminus \hat{U}$.
If $\hat{U}\subseteq V^X_C$, a similar argument to \textsc{Include} defines $k$ clues for $C$, which guarantees the uniqueness of the solution. We then consider the case when $\hat{U}$ contains a vertex outside of $V^X_C$. Such a vertex can be $r$ or one in $V_C^{X'} \cup V_C^{\mathrm{clause}}$. If $r\in \hat{U}$, we can remove it by condition \textbf{2} of the claim. If $v'(l)\in \hat{U}\cap V_C^{X'}$, where $l$ is a literal (i.e., $l=x$ or $l=\bar{x}$ for some $x$), we can replace $v'(l)$ by $v(l)$ in $\hat{U}$ by a similar argument to \textsc{Include} again. 
The last case is that some $v_i^{(j)} \in \hat{U} \cap V_C^{\mathrm{clause}}$. Let $l$ be the literal corresponding to $v_i^{(j)}$. Since $U^*$ is a vertex cover, $v(l)$ should be in $U^*$ to cover edge $\{v_i^{(j)}, v(l)\}$, and $v(\bar{l}) \not\in U^*$ by condition \textbf{3} of the claim.  
Conversely, $v(\bar{l}) \not\in U^*$ implies that $v(l)\in U^*$ and $v_i^{(j)}\not\in U^*$ by condition \textbf{4} of the claim; we can replace $\tilde{U}$ by $\hat{U}\coloneqq\hat{U}\setminus \{v_i^{(j)}\}\cup \{v(\bar{l})\}$, and only-if direction holds. By these, the if-and-only-if relation holds also for the \textsc{Exclude} model.  

We consider the complexity of \pauvc on bipartite graphs.
As observed in \Cref{thm:mix-and-exclude}, it suffices to consider \textsc{Include} and \textsc{Mixed} models.
For bipartite graphs, \textsc{Vertex Cover} is equivalent to the problem of finding a maximum cardinality matching by K\"{o}nig's theorem (see Chapter 2 in \cite{Diestel}), which can be computed in polynomial time by the Hopcroft-Karp algorithm~\cite{HopcroftK73}.
Combining this fact and \Cref{cor:unique_alg}, \pauvc belongs to NP under both models.

To see the NP-hardness of \pauvc, we first consider the \textsc{Mixed} model.
We perform a polynomial-time reduction from \textsc{Independent Dominating Set} on bipartite graphs.
In \textsc{Independent Dominating Set}, we are given a graph $G$ and an integer $k$ and asked whether $G$ has an independent dominating of size at most $k$.
This problem is NP-complete even on bipartite graphs~\cite{CorneilEt:IDS}.

From a bipartite graph $G$ for \textsc{Independent Dominating Set}, we construct a graph $G'$ by adding, for each vertex $v \in V(G)$, a new vertex $v'$ and an edge between $v$ and $v'$.
 We denote by $V' = V(G')\setminus V(G)$ the set of added vertices to $G$ and $E' = E(G')\setminus E(G)$ the set of added edges to $G$.
 Clearly, $G'$ is bipartite and $E'$ is a perfect matching of $G'$.
 By K\"{o}nig's theorem, we have $\tau(G') = |E'| = |V(G)|$.
 This also implies the following observation.

 \begin{observation}\label{obs:bipartite:endvertex}
     Let $U$ be an arbitrary minimum vertex cover of $G'$.
     Then, for any edge $e \in E'$, exactly one end vertex of $e$ belongs to $U$.
 \end{observation}

Now, we show that $G$ has an independent dominating set of size at most $k$ if and only if $G$ has a feasible pre-assignment of size at most $k$.

\begin{lemma}\label{lem:NPhard1}
    If $G$ has an independent dominating set $D$ of size at most $k$, then $G'$ has a feasible pre-assignment $(\tilde{U}_{\rm in}, \tilde{U}_{\rm ex})$ of size at most $k$.
\end{lemma}

\begin{proof}
    Let $D$ be an independent dominating set of size at most $k$ in $G$.
    We show that $(\tilde{U}_{\rm in}, \tilde{U}_{\rm ex}) \coloneqq (\emptyset, D)$ is a feasible pre-assignment for $G'$, that is, $G'$ has a unique minimum vertex cover $U^*$ such that $U^* \subseteq V(G') \setminus D$.

    Let $U^* \coloneqq (V(G)\setminus D) \cup (V' \cap N_{G'}(D))$.
    Clearly, $U^* \subseteq V(G') \setminus D$.
    We first show that $U^*$ is a vertex cover of $G'$.
    Since $D$ is an independent set of $G$, $V(G)\setminus D$ is a vertex cover of $G$, meaning that $V(G)\setminus D$ covers all edges in $E(G)$.
    Thus, edges not covered by $V(G)\setminus D$ are contained in $E'$ and incident to $D$, which can be covered by $V' \cap N_{G'}(D)$.
    Therefore, $U^*$ is a vertex cover of $G'$.
    Moreover,
    \begin{align*}
        |U^*| &= |(V(G) \setminus D) \cup (V' \cap N_{G'}(D))| = |V(G) \setminus D| + |V' \cap N_{G'}(D)|= |V(G)|.
    \end{align*}
    This concludes that $U^*$ is a minimum vertex cover of $G'$, as $\tau(G') = |V(G)|$.

    We next show that $U^*$ is the unique minimum vertex cover of $G'$ under the pre-assignment $(\emptyset, D)$.
    Suppose to the contrary that $G'$ has another minimum vertex cover $U \neq U^*$ such that $U \subseteq V(G') \setminus D$.     
    As $U \neq U^*$, by~\Cref{obs:bipartite:endvertex}, there is a vertex $v' \in V'$ with $v' \in U \setminus U^*$.
    Let $v \in V(G') \setminus V'$ be the unique neighbor of $v'$ in $G'$.
    By~\Cref{obs:bipartite:endvertex}, we have $v \in U^*$, which implies that $v \notin D$.
    Since $D$ is a dominating set of $G$, $v$ has a neighbor $w \in N_{G'}(v)$ in $G'$ that is contained in $D$.
    As $U \subseteq V(G') \setminus D$, $w$ is not contained in $U$.
    This implies that the edge $\{v, w\}$ is not covered by $U$, contradicting the fact that $U$ is a vertex cover of $G'$.
\end{proof}


Next, we show the converse of Lemma~\ref{lem:NPhard1}. 
\begin{lemma}\label{lem:conv}
    If $G'$ has a feasible pre-assignment of size at most $k$, then $G$ has an independent dominating set of size at most $k$.
\end{lemma}

To prove \Cref{lem:conv}, we need several auxiliary lemmas.

\begin{lemma}\label{lem:bipartite:deg1}
    Let $v ' \in V'$ and let $v$ be the unique neighbor of $v'$ in $G'$.
    If there is a feasible pre-assignment $(\tilde{U}_{\rm in}, \tilde{U}_{\rm ex})$ with $v' \in \tilde{U}_{\rm in}$ {\rm(}resp. $v' \in \tilde{U}_{\rm ex}${\rm)} for $G'$, then $(\tilde{U}_{\rm in} \setminus \{v'\}, \tilde{U}_{\rm ex} \cup \{v\})$ {\rm(}resp. $(\tilde{U}_{\rm in} \cup \{v\}, \tilde{U}_{\rm ex} \setminus \{v'\})${\rm)} is a feasible pre-assignment for $G'$.
\end{lemma}
\begin{proof}
    Suppose that $v' \in \tilde{U}_{\rm in}$.
    Let $U^*$ be the unique minimum vertex cover of $G'$ such that $\tilde{U}_{\rm in} \subseteq U^* \subseteq V(G) \setminus \tilde{U}_{\rm ex}$.
    As $v' \in U^*$, by~\Cref{obs:bipartite:endvertex}, it holds that $v \notin U^*$.
    This implies $\tilde{U}_{\rm in} \setminus \{v'\} \subseteq U^* \subseteq V(G) \setminus (\tilde{U}_{\rm ex} \cup \{v\})$.
    Thus, it suffices to show that $U^*$ is the unique minimum vertex cover of $G'$ such that
    $\tilde{U}_{\rm in} \setminus \{v'\} \subseteq U^* \subseteq V(G) \setminus (\tilde{U}_{\rm ex} \cup \{v\})$.
    Suppose for contradiction that $G'$ has another minimum vertex cover $U$ with $U \neq U^*$ such that $\tilde{U}_{\rm in} \setminus \{v'\} \subseteq U \subseteq V(G) \setminus (\tilde{U}_{\rm ex} \cup \{v\})$.
    As $v \notin U$, we have $v' \in U$.
    This implies that $\tilde{U}_{\rm in} \subseteq U$.
    Thus, $\tilde{U}_{\rm in} \subseteq U \subseteq V(G') \setminus \tilde{U}_{\rm ex}$, contradicting the uniqueness of $U^*$.

    The case of $v' \in \tilde{U}_{\rm ex}$ is analogous, which is omitted here.
\end{proof}

By repeatedly applying Lemma~\ref{lem:bipartite:deg1}, the following corollary.

\begin{corollary}\label{cor:bipartite:canonical}
    If $G'$ has a feasible pre-assignment of size at most $k$, then it has a feasible pre-assignment $(\tilde{U}_{\rm in}, \tilde{U}_{\rm ex})$ such that $\tilde{U}_{\rm in} \cup \tilde{U}_{\rm ex} \subseteq V(G') \setminus V'$.
\end{corollary}

Now, we are ready to prove \Cref{lem:conv}.

\begin{proof}[Proof of \Cref{lem:conv}]
    Let $(\tilde{U}_{\rm in}, \tilde{U}_{\rm ex})$ be a feasible pre-assignment of size at most $k$.
    By~\Cref{cor:bipartite:canonical}, we assume that $\tilde{U}_{\rm in} \cup \tilde{U}_{\rm ex} \subseteq V(G') \setminus V'$.
    Let $U^*$ be a vertex cover of $G'$ such that $\tilde{U}_{\rm in} \subseteq U^* \subseteq V(G') \setminus \tilde{U}_{\rm ex}$.
    As $\tilde{U}_{\rm ex} \subseteq V(G') \setminus V' = V(G)$, $\tilde{U}_{\rm ex}$ is an independent set of $G$.  
    Let $D$ be an arbitrary maximal independent set of $G$ that includes $\tilde{U}_{\rm ex}$.
    Since every maximal independent set of $G$ is also an independent dominating set of $G$, it suffices to show that $|D| \le k$.
    Since $D$ is an independent set including $\tilde{U}_{\rm ex}$, $D$ does not contain any vertex in $N_G(\tilde{U}_{\rm ex})$, i.e., $D \subseteq V(G) \setminus N_G(\tilde{U}_{\rm ex})$.
    As we shall show later, it holds that $V(G) = \tilde{U}_{\rm in} \cup \tilde{U}_{\rm ex} \cup N_G(\tilde{U}_{\rm ex})$.
    Thus, 
    \begin{align*}
        |D|
        &\leq |V(G) \setminus N_G(\tilde{U}_{\rm ex})|
        \leq |\tilde{U}_{\rm in} \cup \tilde{U}_{\rm ex}|\leq k.
    \end{align*}
    Therefore, $G$ has an independent dominating set $D$ of size at most $k$.

    To complete the proof, it remains to prove that $V(G) = \tilde{U}_{\rm in} \cup \tilde{U}_{\rm ex} \cup N_G(\tilde{U}_{\rm ex})$.
    Suppose to the contrary that there exists a vertex $v \in V(G)$ with $v \notin \tilde{U}_{\rm in} \cup \tilde{U}_{\rm ex} \cup N_G(\tilde{U}_{\rm ex})$.
    Let $v'$ be the unique vertex of $V'$ adjacent to $v$ in $G'$.
    
    Suppose that $v \not\in U^*$.
    Then we have $v' \in U^*$ due to \Cref{obs:bipartite:endvertex}.
    By the assumption, $v \notin \tilde{U}_{\rm ex}$.
    Moreover, as $\tilde{U}_{\rm in} \cup \tilde{U}_{\rm ex} \subseteq V(G') \setminus V' = V(G)$, we have $v' \notin \tilde{U}_{\rm in} \cup \tilde{U}_{\rm ex}$.
    Thus, $U \coloneqq (U^*\setminus\{v'\})\cup\{v\}$ is a vertex cover of $G'$ satisfying $\tilde{U}_{\rm in} \subseteq U \subseteq V(G') \setminus \tilde{U}_{\rm ex}$, which contradicts the uniqueness of $U^*$.
    
    Suppose otherwise (i.e., $v \in U^*$). 
    In this case, we can assume that all the vertices $w \in V(G)$ with $w \notin \tilde{U}_{\rm in} \cup \tilde{U}_{\rm ex} \cup N_G(X_{\rm ex})$ satisfies $w \in U^*$, as otherwise we derive a similar contradiction in the previous case.
    We first claim that $U \coloneqq (U^*\setminus\{v\})\cup\{v'\}$ is a vertex cover of $G'$.  
    As $U^*$ is a vertex cover of $G'$, all the edges not incident to $v$ are covered by $U$.
    Thus, to prove that $U$ is a vertex cover of $G'$, it suffices to show that $N_G(v) \subseteq U^*$.
    Suppose that $w \in N_G(v)$ does not belong to $U^*$.
    As $v \notin N_{G}(\tilde{U}_{\rm ex})$, we have $w \notin \tilde{U}_{\rm ex}$.
    Moreover, as $w \notin U^*$, we have $w \notin \tilde{U}_{\rm in} \cup N_G(\tilde{U}_{\rm ex})$.
    Thus, $w \notin \tilde{U}_{\rm in} \cup \tilde{U}_{\rm ex} \cup N_G(\tilde{U}_{\rm ex})$, contradicting the above assumption.
    Hence, $U$ is a vertex cover of $G'$.
    
    As $\tilde{U}_{\rm in} \cup \tilde{U}_{\rm ex} \subseteq V(G)$, $v' \notin \tilde{U}_{\rm ex}$.
    Moreover, as $v \notin \tilde{U}_{\rm in}$, it holds that $\tilde{U}_{\rm in} \subseteq U \subseteq V(G') \setminus \tilde{U}_{\rm ex}$, which contradicts the uniqueness of $U^*$.

    This completes the proof.
\end{proof}

By \Cref{lem:NPhard1,lem:conv}, \pauvc is NP-hard.
\begin{theorem}\label{thm:NPC}
    {\rm PAU-VC} for bipartite graphs is NP-complete under the \textsc{Mixed} and \textsc{Exclude} models.
\end{theorem}

By modifying the proof of Theorem \ref{thm:NPC}, we can show that the hardness also holds for \textsc{Include} model. 
As shown in~\Cref{lem:NPhard1,lem:conv}, $G$ has an independent dominating set of size at most $k$ if and only if $G'$ has a feasible pre-assignment $(\tilde{U}_{\rm in}, \tilde{U}_{\rm ex})$ of size at most $k$ (under the \textsc{Mixed} model).
From this pre-assignment, we define
\begin{align*}
    \tilde{U} \coloneqq & \tilde{U}_{\rm in} \cup\{v \in V(G) : v' \in \tilde{U}_{\rm ex} \cap V'\} \cup \{v' \in V' : v \in \tilde{U}_{\rm ex} \cap V(G)\}.
\end{align*}
Let $U^*$ be a minimum vertex cover of $G'$ such that $\tilde{U}_{\rm in} \subseteq U^* \subseteq V(G') \setminus \tilde{U}_{\rm ex}$.
By~\Cref{obs:bipartite:endvertex}, all vertices in $\{v \in V(G) : v' \in \tilde{U}_{\rm ex} \cap V'\}$ and all vertices in $\{v' \in V' : v \in \tilde{U}_{\rm ex} \cap V(G)\}$ belong to $U^*$.
Thus, $\tilde{U} \subseteq U^*$.
Suppose that there is another minimum vertex cover $U$ of $G'$ such that $\tilde{U} \subseteq U$.
Again, by~\Cref{obs:bipartite:endvertex}, all vertices in $\tilde{U}_{\rm ex}$ do not belong to $U$.
Thus, we have $\tilde{U}_{\rm in} \subseteq U \subseteq V(G') \setminus \tilde{U}_{\rm ex}$, contradicting the uniqueness of $U^*$.

As a feasible pre-assignment under the \textsc{Include} mode can be seen as a feasible pre-assignment on the \textsc{Mixed} model (as we have seen in~\Cref{thm:mix-and-exclude}), we can conclude that $G$ has an independent dominating set of size at most $k$ if and only if $G'$ has a feasible pre-assignment of size at most $k$ under the \textsc{Include} model.

\begin{theorem}\label{thm:NPC2}
    {\rm PAU-VC} for bipartite graphs is NP-complete under the \textsc{Include} model.
\end{theorem}

\section{Exact Algorithms}

\subsection{General graphs}
In this section, we present exact algorithms for \pauvc for general graphs. 
We first see exact exponential-time algorithms, and then see FPT algorithms for vertex cover number. 
As shown in \Cref{thm:mix-and-exclude}, we only consider \textsc{Include} and \textsc{Exclude} models.

\subsubsection{Exponential-time algorithms}
We first present an exact algorithm that utilizes an algorithm for \textsc{Unique Vertex Cover}~(\textsc{UVC}). 
Let $G$ be a graph with $n$ vertices.
We first fix a subset of vertices $U \subseteq V(G)$ for a pre-assignment, and then check if $G$ has a unique minimum vertex cover under the pre-assignment $U$.
This can be done by transforming $G$ into $G'$ as in~\Cref{cor:unique_alg}: $G' \coloneqq G - U$ for the \textsc{Include} model and $G' \coloneqq G - U - N(U)$ for the \textsc{Exclude} model.
By applying this procedure for all subsets $U$, we can determine the answer of \pauvc for $G$. 
\begin{algorithm}[H]
\begin{algorithmic}
\caption{Algorithm using \textsc{UVC} routine}
\label{alg:simple2}
\For{$k=0,1,\ldots,n$}
\ForAll{$U\subseteq V$ with $|U|=k$}
    \State $G' \gets$ the graph obtained from $G$ under pre-assignment on $U$
    \If{UVC($G'$) = true} \label{alg:line1}
    \State \textbf{Return} $U$
    \EndIf
\EndFor
\EndFor 
\end{algorithmic}
\end{algorithm}
The running time depends on the algorithm to solve \textsc{UVC}. We abuse the notation UVC($G'$) in \Cref{alg:simple2}: UVC($G'$) returns true if and only if $G'$ has a unique minimum vertex cover of size $\tau(G) - |U|$ for the \textsc{Include Model} and size $\tau(G) - |N(U)|$ for the \textsc{Exclude} model.
By \Cref{lem:unique_alg}, \textsc{UVC} can be solved in the same exponential order of running time as \textsc{Vertex Cover}. Here, let $O^*(\alpha^n)$ be the running time of an exact exponential-time algorithm for \textsc{Vertex Cover}. Note that Algorithm \ref{alg:simple2} applies the \textsc{Vertex Cover} algorithm for $G'$, which has at most $n-|U|$ vertices in any pre-assignment model. Thus, the total running time of Algorithm \ref{alg:simple2} is estimated as 
\[
\sum_{k=0}^n \binom{n}{k} O^*(\alpha^{n-k}) = \sum_{k=0}^n \binom{n}{k} O^*(\alpha^{k}) = O^*({(\alpha+1)}^{n}), 
\]
by the binomial theorem. Since the current fastest exact exponential-time algorithm for \textsc{Vertex Cover} runs in $O(1.1996^n)$ time~\cite{XiaoNagamochi2017}, Algorithm \ref{alg:simple2} runs in $O(2.1996^n)$ time. 

\begin{theorem}
  {\rm PAU-VC} for any model can be solved in $O(2.1996^n)$ time.    
\end{theorem}

\subsubsection{FPT algorithms parameterized by vertex cover number}
We now present FPT algorithms parameterized by vertex cover number. 
Throughout this subsection, we simply write $\tau \coloneqq \tau(G)$.
Contrary to the previous subsection, algorithms for \textsc{Include} and \textsc{Exclude} work differently. 
\subsection{\textsc{Include} Model}
We first present the algorithm for \textsc{Include}, which is easier to understand. 
The idea of the algorithm is as follows. Let $G$ be a graph. We first enumerate all the minimum vertex covers of $G$. 
Then, for every optimal vertex cover $U^*$, we find a minimum feasible pre-assignment. Namely, we fix a subset $U\subseteq U^*$ as \textsc{Include} vertices, and then check the uniqueness under $U$ as Algorithm \ref{alg:simple2}. We formally describe the algorithm in Algorithm~\ref{alg:fptIn}.  

\begin{algorithm}[H]
\begin{algorithmic}[1]
\caption{FPT algorithm for \textsc{Include}}
\label{alg:fptIn}
\State Enumerate all the optimal vertex covers of $G$, and let $\mathcal{U}$ be the collection of them.   
\ForAll{$U^*\in \mathcal{U}$}
\State $\mathrm{Sol}(U^*) \gets U^*$ \Comment{Current min. pre-asgmt. for $U^*$}
\ForAll{$U\subseteq U^*$ with $|U|< \tau$}
    \State $G' \gets$ the graph obtained from $G$ under pre-assignment on $U$
    \If{$|U| < |\mathrm{Sol}(U^*)|$ and UVC($G'$) = true} 
    \State $\mathrm{Sol}(U^*)\gets U$
    \EndIf
    \EndFor
\EndFor 
\State \textbf{Return} $\mathrm{Sol}(U^*)$ with the minimum size for $U^* \in \mathcal{U}$
\end{algorithmic}
\end{algorithm}
Inside of the algorithm, we invoke an algorithm for computing a minimum vertex cover of $G'$ parameterized by vertex cover number $\tau(G')$, and $O^*(\beta^{\tau(G')})$ denotes the running time. 
Since $G'$ has a vertex cover at most $\tau-|U|$, the running time of Algorithm \ref{alg:fptIn} is estimated as follows: 
\begin{align}
\begin{aligned}
\sum_{U^*\in \mathcal{U}} \sum_{U \subseteq U^*} O^*(\beta^{\tau-|U|}) &= |\mathcal{U}|\sum_{k=0}^{\tau}\binom{\tau}{k} O^*(\beta^{\tau-k}) =  O^*(|\mathcal{U}|{(\beta+1)}^{\tau}). \label{eq:eq1}
\end{aligned}
\end{align}
    
Thus bounding $|\mathcal{U}|$ by a function of $\tau$ is essential. 
Actually, $|\mathcal{U}|$ is at most $2^{\tau}$ by the following reason. 
For $\mathcal{U}$, we consider classifying the optimal vertex covers of $G$ into two categories for a non-isolated vertex $v \in V(G)$: (1) $\mathcal{U}(v)$ is the set of optimal vertex covers of $G$ containing $v$, and (2) $\mathcal{U}({\bar{v}})$ is the set of optimal vertex covers of $G$ not containing $v$. Note that any vertex cover in $\mathcal{U}({\bar{v}})$ contains all the vertices in $N(v)$ by the requirement of a vertex cover; a vertex cover in $\mathcal{U}({v})$ (resp., $\mathcal{U}(\bar{v})$) reserves $v$ (resp., $N(v)$) for vertex cover and can include at most $\tau-1$ (resp., $\tau-|N(v)|$) vertices. We can further classify $\mathcal{U}({v})$ (resp., $\mathcal{U}(\bar{v})$) into $\mathcal{U}(vv')$ and $\mathcal{U}(v\bar{v'})$ (resp., $\mathcal{U}(\bar{v}v')$ and $\mathcal{U}(\bar{v}\bar{v'})$) similarly, which gives a branching with depth at most $\tau$. Since every optimal vertex cover is classified into a leaf of the branching tree, and thus $|\mathcal{U}|\le 2^{\tau}$ holds; the above running time is $O^*({(2\beta+2)}^{\tau})$.  

Although it is impossible to improve bound $2^{\tau}$ on the number of minimum vertex covers of $G$,\footnote{If $G$ is a disjoint union of $n/2$ isolated edges, $G$ has exactly $2^{\tau}$ minimum vertex covers.} we can still improve the running time. 
The idea is to adopt a smaller set of optimal vertex covers instead of the whole set of the optimal vertex covers. Here, we focus on $\mathcal{U}(\bv)$, where $\bv$ is a sequence of symbols representing a vertex or its negation, as discussed above. Since $\bv$ has a history of branching, it reserves a set of vertices as a part of an optimal vertex cover in $\mathcal{U}(\bv)$. Let $U(\bv)$ denote the corresponding vertex set (i.e., the vertices appearing positively $\bv$). Then, $\tau-|U(\bv)|$ vertices are needed to cover the edges in $G - U(\bv)$. 

We now assume that $G - U(\bv)$ forms a collection of isolated edges. In such a case, we can exploit the structure of an optimal vertex cover without branching to leaves by the following argument. To uniquify a minimum vertex cover of $G$, we need to specify exactly one of two endpoints of each isolated edge. In Algorithm \ref{alg:fptIn}, we fix a target vertex cover $U^*$ at line 2, choose a subset of $U^*$ for pre-assignment \textsc{Include} at line 4, and check if it ensures the uniqueness of an optimal vertex cover. Instead of fixing a target vertex cover $U^*$, we focus on $U(\bv)$ with $k'$ isolated edges, where $|U(\bv)|+k'=\tau$ holds (as otherwise $G$ has no vertex cover extending $U(\bv)$ of size $\tau$). 
%
That is, a pre-assignment with size $k~(\ge k')$ on $U(\bv)$ with $k'$ isolated edges must form $k'$ end points of the isolated edges (i.e., one of $2^{k'}$ choices) and a $(k-k')$-subset of $U(\bv)$ (one of $\binom{|U(\bv)|}{k-k'}$). Thus, for all the pre-assignments on $U(\bv)$ plus $k'$ isolated edges, the uniqueness check takes  
\begin{align*}
 & 2^{k'}\sum_{k''=0}^{\tau-k'}\binom{\tau-k'}{k''} \cdot O^*(\beta^{\tau-k'-k''})=2^{k'} \cdot O^*((\beta+1)^{\tau-k'}) = O^*((\beta+1)^{\tau}).
\end{align*}

To obtain such $G - U(\bv)$ (i.e., a graph consisting of isolated edges), we do branching by $v$ with a degree at least $2$.  
If no vertex with a degree at least $2$ is left, then $G - U(\bv)$ forms a collection of isolated edges. We now estimate the size $f(\tau)$ of a branching tree. 
If $v$ with a degree at least $2$ is in the vertex cover, the remaining graph has a vertex cover with at least $\tau-1$ vertices. If $v$ is not in the vertex cover, $N(v)$ should be included in a vertex cover; the remaining graph has a vertex cover with at least $\tau-2$.  
Thus we have
\[
 f(\tau) \le f(\tau-1) + f(\tau-2). 
\]
This recurrence inequality leads to $f(\tau)\le \phi^\tau$, where $\phi=(1+\sqrt{5})/2$ is the golden ratio.  Overall, the running time of the revised Algorithm \ref{alg:fptIn} becomes $O^*({(\phi(\beta+1))}^{\tau})$. 
Since the current fastest vertex cover algorithm runs in $O^*(1.2738^{\tau})$ time~\cite{ChenKX10:TCS:Improved},  
we can solve \pauvc under \textsc{Include} model in $O^*(3.6791^{\tau})$ time. 
\begin{theorem}
  {\rm PAU-VC} under \textsc{Include} model can be solved in $O^*(3.6791^{\tau})$ time.    
\end{theorem}

\subsection{\textsc{Exclude} Model}
We next consider \textsc{Exclude} model. A straightforward extension of Algorithm \ref{alg:fptIn} changes line 4 as:
\begin{equation}\label{change}
\textbf{ for all }  \tilde{U} \subseteq V\setminus U^*  \text{ with } |\tilde{U}|=k \textbf{ do }
\end{equation}
In line 5, $G'$ is defined from $\tilde{U}$ as in the \textsc{Exclude} model.
Although this change still guarantees that the algorithm works correctly, 
it affects the running time. Equation (\ref{eq:eq1}) becomes 
\[
|\mathcal{U}|\sum_{k=0}^{n-\tau} \binom{n-\tau}{k} O^*(\beta^{\tau-k}) = O^*(|\mathcal{U}| \beta^{\tau} (1+1/\beta)^{n-\tau}), 
\]
which is no longer to be fixed parameter tractable with respect to $\tau$. To circumvent the problem, we show the following lemma. 
\begin{lemma}\label{lem:inout}
Suppose two subsets $\tilde{U}_1$ and $\tilde{U}_2$ of $V(G)$ have the same neighbors, i.e., $N(\tilde{U}_1)=N(\tilde{U}_2)$. If a set $U\subseteq V\setminus \tilde{U}_1$ is an optimal vertex cover of $G$, $U\cap \tilde{U}_2=\emptyset$ holds. 
\end{lemma}

\begin{proof}
Note that $N(\tilde{U}_1)(=N(\tilde{U}_2))$ and $U_1\cup U_2$ are disjoint.  
Suppose that a set $U$ with $\tilde{U}_1\cap U=\emptyset$ is an optimal vertex cover of $G$. 
Then, $N(\tilde{U}_1) \subseteq U$ holds, because otherwise the edges incident with $\tilde{U}_1$ are never covered. 
Since $N(\tilde{U}_1)~(=N(\tilde{U}_2))$ covers also the edges incident with $\tilde{U}_2$, $U\setminus \tilde{U}_2~(\supseteq N(\tilde{U}_2))$ is also a vertex cover. By the optimality of $U$, $U\setminus \tilde{U}_2$ must be $U$, which implies $U\cap \tilde{U}_2=\emptyset$. 
\end{proof}

Lemma \ref{lem:inout} implies that for two subsets $\tilde{U}_1$ and $\tilde{U}_2$ of $V(G)$ with $N(\tilde{U}_1)=N(\tilde{U}_2)$, excluding $\tilde{U}_1$ and excluding $\tilde{U}_2$ have the same effect as including $N(\tilde{U}_1)~(=N(\tilde{U}_2))$ for optimal vertex covers. Namely, for an optimal vertex cover $U$, the ``cost'' of including $U'~(\subseteq U)$ is interpreted as the minimum $|\tilde{U}|$ such that $U'\subseteq N(\tilde{U})$. 
The notion of costs enables us to transform a pre-assignment of \textsc{Include} into one of \textsc{Exclude}; we can utilize the same framework of Algorithm \ref{alg:fptIn} by installing a device to compute the costs. We thus consider obtaining the costs efficiently. The cost of $U'$ is the optimal value of the following set cover problem: The elements to be covered are the vertices in $U'$, and a set to cover elements is $u\in V\setminus U$, which consists of $N(u)$. By executing a standard dynamic programming algorithm 
for the set cover problem for the ground set $U$ once, 
we obtain the cost to cover every subset of $U$ in $O^*(2^{|U|})$ time (also see e.g.,~\cite{kratsch2010exact}). 
\begin{algorithm}[H]
\begin{algorithmic}[1]
\caption{Algorithm for \textsc{Set Cover}}
\label{alg:setcover}
\Require $(X,\mathcal{S})$, where $X=\{1,2,\ldots,n\}$ is a ground set (i.e., the set of elements to be covered) and $\mathcal{S}=\{S_1,S_2,\ldots,S_m\}$ is a collection of subsets of $X$. 
\Ensure For $S\subseteq X$, $\cost[S;j]$ is the minimum cardinality of a subset of $\{S_1,S_2,\ldots,S_j\}$ that covers $S\subseteq X$. If $S$ cannot be covered by $\{S_1,S_2,\ldots,S_j\}$, $\cost$ is set to $\infty$. 
\ForAll{$S\subseteq X$} 
\For{$j=1,\ldots,n$}
    \State $\cost[S;j] \gets \min\{\cost[S;j-1],\cost[S\setminus S_{j};j-1]+1\}$
\EndFor
\EndFor
\end{algorithmic}
\end{algorithm}
Note that Algorithm \ref{alg:setcover} for $(X,\mathcal{S})$ computes the optimal cost to cover $S$ (denoted $\cost[S]$) but can easily modified to compute an optimal solution $\SC[S]$ for every $S\subseteq X$ in the same running time.

Algorithm \ref{alg:fptEx} is the algorithm for \textsc{Exclude} model. It works in a similar framework to Algorithm \ref{alg:fptIn}, but it equips the device to transform \textsc{Include} pre-assignment into \textsc{Exclude} pre-assignment at line~3, which solves the set cover problem $(U^*,\{N(u) \mid u \in V(G) \setminus U^*\})$ in advance. Line~4 fixes a subset $U$ as vertices to be included in $U^*$ by pre-assigning some vertices in $V(G) \setminus U^*$ \textsc{Exclude}. To link $U$ to such vertices in $V(G) \setminus U^*$, we use the solutions of the set cover problem $(U^*,\{N(u) \mid u \in V(G) \setminus U^*\})$, whose correctness is guaranteed by Lemma \ref{lem:inout}.  The other structures are essentially equivalent to those of Algorithm \ref{alg:fptIn}.

\begin{algorithm}[H]
\begin{algorithmic}[1]
\caption{FPT algorithm for \textsc{Exclude}}
\label{alg:fptEx}
\State Enumerate all the optimal vertex covers. Let $\mathcal{U}$ be the set of them. 
\ForAll{optimal vertex cover $U^*$}
\State Apply Algorithm \ref{alg:setcover} for $(U^*,\{N(u) \mid u \in V(G) \setminus U^*\})$. 
\State $\mathrm{Sol}(U^*) \gets \SC(U^*)$
\ForAll{$U\subseteq U^*$ with $|\SC(U)|< |\mathrm{Sol}(U^*)|$}
    \State $G' \gets$ the graph obtained from $G$ under pre-assignment on $\SC(U)$ (\textsc{Exclude})
    \If{$|U| < |\mathrm{Sol}(U^*)|$ and UVC($G'$) = true} 
    \State $\mathrm{Sol}(U^*)\gets\SC(U)$
    \EndIf
    \EndFor
\EndFor 
\State \textbf{Return} $\mathrm{Sol}(U^*)$ with the minimum size for $U^* \in \mathcal{U}$
\end{algorithmic}
\end{algorithm}
The running time is estimated as follows. We apply Algorithm \ref{alg:setcover} for each $U^*$, which takes $O^*(2^{\tau})$ time. For each $U\subseteq U^*$ with $|\SC(U)|< |\mathrm{Sol}(U^*)|$, we obtain $G'$, whose vertex cover has size at most $|U^*|-|U|$. By these, we can estimate the running time as 
\begin{align*}
&\sum_{U^*\in \mathcal{U}} \left(O^*(2^{\tau})+\sum_{U\in U^*} O^*(\beta^{\tau-|U|})\right) \\
=&\ |\mathcal{U}|\left(O^*(2^{\tau})+\sum_{k=0}^{\tau} \binom{\tau}{k} O^*(\beta^{k})\right) =O^*(|\mathcal{U}|{(\beta+1)}^{\tau}), 
\end{align*}
because $\beta>1$. By $|\mathcal{U}|\le 2^{\tau}$, we can see that \pauvc under \textsc{Exclude} model can be solved in $O^*(4.5476^{\tau})$ time. 

Here, we further try to improve the running time as \textsc{Include} model.  
Instead of an optimal vertex cover $U^*$, we focus on $U(\bv)$ in the branching argument of the previous subsection, where $G - U(\bv)$ forms a collection of $k'$ edges. 
Since almost the same argument holds, lines 5$-$7 are executed at most $2^{|U|}$, where $U$ is an 
\textsc{Include} pre-assignment (i.e., an implicit \textsc{Exclude} pre-assignment) for an optimal vertex cover $U(\bv)$ plus a set of $k'$ end vertices of the isolated edges in $G - U(\bv)$. 
Thus, what we need to consider newly for \textsc{Exclude} model is how we efficiently get the set cover solution (i.e., vertices to be excluded) of every $U$ for $U(\bv)$. 

We explain how we resolve this. 
Suppose that $K(\bv)$ is the set of isolated edges in $G - U(\bv)$, and we then define the following set cover problem: The elements to be covered are the vertices in $U(\bv)$, and a set to cover elements is $u\in V(G) \setminus (U(\bv)\cup V(K(\bv)))$, which consists of $N(u)$. By applying Algorithm \ref{alg:setcover} for $(U(\bv),\{N(u)\mid u\in V(G) \setminus (U(\bv)\cup V(K(\bv)))\})$, we obtain for every $U'\subseteq U(\bv)$, $\SC'(U')$ as a minimum set to cover $U'$, where we use $\SC'$ instead of $\SC$ to distinguish. 
Let $U^*$ be an optimal vertex cover including $U(\bv)$, and let $K_1:=U^*\cap V(K(\bv))$ and $K_2:=V(K(\bv))\setminus U^*$. 
We claim that for $U\subseteq U^*$ with $N(K_2)\subseteq U$, $\SC'(U\setminus N(K_2))\cup K_2$ is an optimal solution for $U$ of the set cover problem $(U^*,\{N(u)\mid u\in V(G) \setminus U^*\})$. In fact, an optimal solution for $U$ must contain $K_2$ to cover $K_1$, which can be covered only by $K_2$. 
The remaining vertices of the optimal solution are in $V(G) \setminus (U(\bv)\cup V(K(\bv)))$ and optimally cover $U\setminus N(K_2)$; this is the requirement of $\SC'(U\setminus N(K_2))$. Thus, the solutions of the set cover problems can be obtained in the same running time. 
Thus, we can apply the same branching rule as \textsc{Include} model, which leads to Theorem \ref{thm:fptex}.  

\begin{theorem}\label{thm:fptex}
{\rm PAU-VC} under \textsc{Exclude} or \textsc{Mixed} model can be solved in $O^*(3.6791^{\tau})$ time.    
\end{theorem}

Since the vertex cover number of a bipartite graph is at most $n/2$, 
{\rm PAU-VC} for bipartite graphs under any model can be solved in $O(1.9181^{n})$ time. 

\subsection{Trees}
In this section, we give an exponential upper bound $O(1.4143^n)$ of the time complexity of \pauvc for trees. 

We first give a generic recursive algorithm for \pauvc of trees in Algorithm \ref{alg:tree}. 
Here, ``generic'' means that it describes the common actions for \textsc{Include} and \textsc{Exclude} models. 
Because of this, lines $10$ and $13$ are described in ambiguous ways, which will be given later.
The algorithm is so-called a divide and conquer algorithm. Lines from 2 to 5 handle the base cases. The following easy lemma might be useful. 
\begin{lemma}
A star with at least two leaves (i.e., $K_{1,l}$ with $l\ge 2$) has the unique optimal vertex cover. 
A star with one leaf (i.e., a single edge) has two optimal vertex covers. 
\end{lemma}
At line 9, choose a dividing point $v \in V(T)$ as a pre-assigned vertex, and after that, apply the algorithm to the divided components. Here, notice that the chosen vertex may not be included in any minimum feasible pre-assignment for $T$.
This is the reason why we try all possible $v$, which guarantee that at least one $v$ is correct, though it causes recursive calls for many subtrees of $T$. By using memoization, we can avoid multiple calls for one subtree; the number of subtrees appearing in the procedure dominates the running time. In the following subsections, we estimate an upper bound on the number of subtrees that may appear during the procedure in Algorithm \ref{alg:tree} for each pre-assignment model. 
\begin{algorithm}
\begin{algorithmic}[1]
\caption{Exact algorithm for trees}
\label{alg:tree}
\Function{PAU-Tree}{$T$} \Comment{Return $(\tau(T),\mathrm{opt}(T))$} 
\State Compute an optimal vertex cover of $T$ 
\If{$T$ has a unique optimal vertex cover}
\State \textbf{Return} $(\tau(T),0)$
\ElsIf{$T$ is a single edge}
\State \textbf{Return} $(1,1)$
\Else
\State $k\gets |V(T)|$
\ForAll{$v\in V(T)$} \Comment{$v$ is chosen for a pre-assignment}
\State Let $T_1,\ldots,T_j$ be the connected components of the graph obtained by pre-assigning $v$. 
\For{$i=1,\ldots,j$}
\State $(a_i,b_i) \gets$ \textsc{PAU-Tree}$(T_i)$
\EndFor
\If{$\sum_{i} a_i$ meets $\tau(T)$}
\State $k'\gets 1+\sum_{i} b_i$ \Comment{``$1$'' comes from $|r|=1$}
\If{$k'<k$}
\State $k\gets k'$
\EndIf
\EndIf
\EndFor
\State \textbf{Return} $(\tau(T),k)$
\EndIf
\EndFunction
\end{algorithmic}
\end{algorithm}


\smallskip

In the \textsc{Include} Model, we slightly change lines $9$ and $10$ in Algorithm \ref{alg:tree} as follows:
\begin{align*}
& \textbf{ for all }  v \in V(T)   \text{ with } d(v)>1 \textbf{ do } \\    
& \hspace{.65cm}\text{Let $T_1,\ldots,T_i$ be the connected components of $T - \{v\}$. }
\end{align*}

\smallskip 

This change is safe by the following lemma. 
\begin{lemma}\label{lem:star}
Let $T$ be a tree, which is not a star. Then, there exists a minimum feasible pre-assignment $\tilde{U}$ for $T$ in the \textsc{Include} model, where $\tilde{U}$ does not contain a leaf vertex. 
\end{lemma}

\begin{proof}
We show that if $\tilde{U}$ contains a leaf vertex, we can construct another minimum feasible pre-assignment $\tilde{U}'$ for $T$ in the \textsc{Include}, 
which does not contain a leaf vertex. We first see the \textsc{Include} model. Suppose that $\tilde{U}$ contains a leaf vertex $u$ and the target vertex cover is $U^*$. Since $u\in U^*$ is a leaf, it covers only an edge $\{u,u'\}$, where $u'$ is the unique neighbor of $u$ and $u'\not\in U^*$. Then, for minimum vertex cover $U^*\cup \{u'\} \setminus \{u\}$, $\tilde{U}\cup \{u'\} \setminus \{u\}$ is a minimum feasible pre-assignment. Note that $u'$ is not a leaf, because $T$ is not a single edge. Thus, we can replace a leaf with a non-leaf vertex in $\tilde{U}$ with keeping the optimality, which shows the lemma's statement for the \textsc{Include} model. 
\end{proof}

The if-condition at line 13 is also changed to $\sum_{i}a_i = \tau(T) - 1$ so as $v$ to be included in the vertex cover.

Now we give an upper bound on the number of subtrees of $n$ vertices that may appear during the procedure in Algorithm \ref{alg:tree}. To make it easy to count subtrees, we fix a root; the number of subtrees is upper bounded by $n$ times an upper bound of the number of rooted subtrees.  
Furthermore, we introduce the notion of isomorphism with a root to reduce multiple counting. 
\begin{definition}\label{def:isomorphic:tree}
Let $T^{(1)}=(V^{(1)},E^{(1)},r^{(1)})$ and $T^{(2)}=(V^{(2)},E^{(2)},r^{(2)})$ be trees rooted at $r^{(1)}$ and $r^{(2)}$, respectively. Then,  
$T^{(1)}$ and $T^{(2)}$ are called \emph{isomorphic with respect to root} if for any pair of $u,v\in V^{(1)}$ there is a bijection $f: V^{(1)}\to V^{(2)}$ such that $\{u,v\}\in E^{(1)}$ if and only if $\{f(u),f(v)\}\in E^{(2)}$ and $f(r^{(1)})=f(r^{(2)})$. 
\end{definition}
For $T=(V,E)$ rooted at $r$, a connected subtree $T'$ rooted at $r$ is called a \emph{rooted I-subtree} of $T$, if $T'$ is $T$ itself, or there exists a non-leaf $v$ such that $T'$ is the connected component with root $r$ of $T - \{v\}$. Note that the graph consisting of only vertex $r$ can be a rooted I-subtree.

\begin{lemma}\label{lem:treeIn}
Any tree rooted at $r$ has $O^*(2^{n/2})~(=O(1.4143^n))$ non-isomorphic rooted I-subtrees rooted at $r$, where $n$ is the number of the vertices.
\end{lemma}
\begin{proof}
Let $R(n)$ be the maximum number of non-isomorphic rooted I-subtrees of any tree rooted at some $r$ with $n$ vertices. We claim that $R(n)\le 2^{n/2}-1$ for all $n\geq 4$, which proves the lemma. 
In a different context, Yoshiwatari et al. show that the same inequality holds for the number of subtrees satisfying a certain property (called AK-rooted subtrees)~\cite[Lemma 5]{yoshiwatari2022winnerc}.  
Note that the notions of rooted I-subtree and AK-rooted subtree are different; a rooted I-subree is obtained by removing vertices, while an AK-rooted subtree is obtained by removing edges. 


To prove the current lemma, we reuse the proof of \cite[Lemma 5]{yoshiwatari2022winnerc}. 
A key is that for small cases ($n\le 4$), the numbers of non-isomorphic rooted I-subtrees are the same as 
non-isomorphic AK-rooted subtrees, which is a building block of the argument. Concretely, what we show here is (1) $R(1)=1, R(2)=1, R(4)=3$, and (2) for $n=3$ in Figure \ref{tree1}, the number of subtrees of Type A is $2$ and that of Type B is $1$; if these hold, exactly the same argument can be applied and we have  $R(n)\le 2^{n/2}-1$ for all $n\geq 4$. 

We first see (1). 
Since $n=1$ is just the root itself, $R(1)=1$. For $n=2$, a tree $T$ with $2$ vertices is a single edge, and a rooted I-subtree of $T$ containing $r$ is $T$ itself; $R(2)=1$. For $n=4$, Figure \ref{tree2} lists all candidates of $T$, and the numbers of the rooted I-subtrees are 3, 2, 2, 1, respectively; we have $R(4)=3$. 
We next see (2). For Type A in Figure \ref{tree1}, the rooted I-subtrees are the tree itself and isolated $r$, and for Type B, a rooted I-subtree is only the tree itself. Since (1) and (2) hold, we can prove this lemma by following the argument in the proof of \cite[Lemma 5]{yoshiwatari2022winnerc}.   
\begin{figure}[t]
\begin{minipage}{0.48\hsize}
 \centering
 \includegraphics[width=0.55\linewidth]{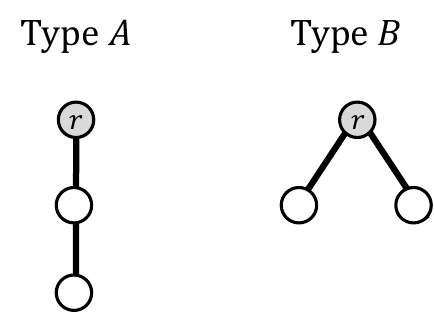}
 \caption{Trees with 3 vertices rooted at $r$}
 \label{tree1}
\end{minipage}
\begin{minipage}{0.48\hsize}
 \centering
 \includegraphics[width=1.05\linewidth]{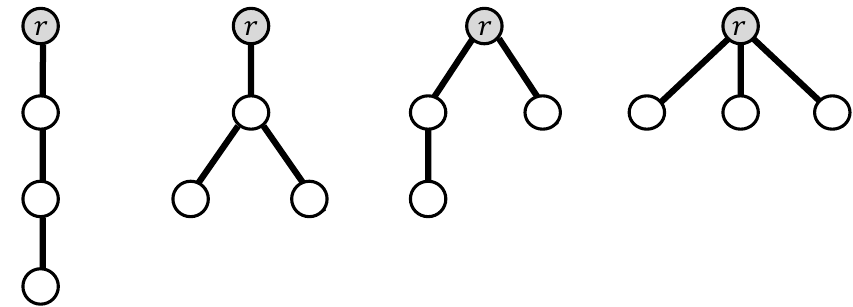}
 \caption{Trees with 4 vertices rooted at $r$}
 \label{tree2}
\end{minipage}
\end{figure}
\end{proof}

Lemma \ref{lem:treeIn} leads to the running time for \textsc{Include} model. By applying a different but similar argument, we achieve the same running time for \textsc{Exclude} model 
\begin{theorem}
{\rm PAU-VC} for trees under any model can be solved in time $O^*(2^{n/2})=O(1.4143^n)$.
\end{theorem}


\bibliography{arxiv}
\end{document}